\documentclass[reqno]{amsart}
\usepackage[all]{xy}
\usepackage{amsmath,amsfonts,amssymb,amsthm,epsfig,amscd,comment,latexsym,psfrag}
\usepackage{nicefrac,xspace,tikz}
\usetikzlibrary{arrows}

\newtheorem{prop}{Proposition}
\newtheorem{theorem}{Theorem}
\newtheorem{cor}{Corollary}

\newtheorem{lemma}{Lemma}

\theoremstyle{definition}
\newtheorem{Remark}{Remark}

\def\sP{\mathsf{P}}
\def\sQ{\mathsf{Q}}

\def\mcP{\mathcal P}
\def\mcQ{\mathcal Q}
\def\mcPb{\bar{\mcP}}

\def\lba{[}
\def\rba{]}

\def\H{\mathcal H}

\def\R{\mathbb R}
\def\mcR{\mathcal R}
\def\Q{\mathbb Q}
\def\Z{\mathbb Z}
\def\Zt{\Z[t,t^{-1}]}

\def\C{\mathbb C}
\def\Ct{\C[t,t^{-1}]}

\def\lra{\longrightarrow}

\def\mbf1{\mathbf 1}
\def\tb1{\widetilde{\mbf1}}
\def\tr{\mbox{tr}}
\def\tp{\bar{p}}

\def\sp{\mathsf p}
\def\sq{\mathsf q}
\def\sa{\mathsf a}
\def\m{\mathbf m}
\def\n{\mathbf n}
\def\gmod{\rm\bf -gmod}
\def\An{A_{n}}

\def\lra{\longrightarrow}
\def\dmod{\rm\bf -mod}

\def\End{\mathrm{End}}
\def\EndHp{\End_{\H'}}

\def\HomHp{\Hom_{\H'}}
\def\HomHpp{\Hom_{'\H'}}

\def\id{\mathrm{id}}

\newcommand{\Hom}{{\rm Hom}}

\psfrag{Qpm}{$Q_{+-}$} \psfrag{Qmp}{$Q_{-+}$}
\psfrag{Ione}{$\quad\mathbf{1}$}
\psfrag{Xiota}{$\iota$}\psfrag{Xiotap}{$\iota'$}
\psfrag{Xgi}{$g_i$}\psfrag{Xgim1}{$g_i^{-1}$}
\psfrag{Xsumi}{\large{$\sum_{i\in I}$}}
\psfrag{Xbeta}{$\overline{\beta}$}
\psfrag{Xgj}{$g_j$} \psfrag{ifinotj}{\large{if $i\not= j$}}
\psfrag{XSnLm}{$S^n_-\otimes \Lambda^m_+$}
\psfrag{XLmSn}{$\Lambda^m_+\otimes S^n_-$}
\psfrag{XLm1Sn1}{$\Lambda^{m-1}_+\otimes S^{n-1}_-$}
\psfrag{XSumkb}{$-\sum_{b=0}^{k-2} (k-b-1)$}

\title{A note on Heisenberg categorification}

\date{}

\begin{document}

\author{Na Wang}
\author{Zhixi Wang}
\author{Ke Wu}
\author{Jie Yang}
\author{Zifeng Yang}

\address{Na Wang\\College of mathematics and information sciences
\\Henan University, Kaifeng, 475001, China\\}
\email{wangnawanda@163.com}

\address{Zhixi Wang\\School of mathematical science\\ Capital Normal University,
Beijing, 100048, China\\}
\email{wangzhx@mail.cnu.edu.cn}

\address{Ke Wu\\School of mathematical science\\Capital Normal University,
Beijing, 100048, China\\}
\email{wuke@mail.cnu.edu.cn}

\address{Jie Yang\\School of mathematical science\\Capital Normal University,
Beijing, 100048, China\\}
\email{yangjie@cnu.edu.cn}

\address{Zifeng Yang\\School of mathematical science\\Capital Normal University,
Beijing, 100048, China}
\email{yangzf@mail.cnu.edu.cn}

\begin{abstract}
A categorification of the Heisenberg algebra is constructed in \cite{K} by Khovanov
using graphical calculus, and left with a conjecture on the isomorphism between the
Heisenberg algebra and Grothendieck ring of the constructed category.
We give a proof of Khovanov's conjectured statement in this paper from
a categorification of some deformed Heisenberg algebra in \cite{WWY}.
\end{abstract}

\maketitle
\setcounter{page}{1}


\section{Introduction}
The classical Heisenberg algebra is the algebra of infinite rank generated by the sequence
$\{a_n\}_{n\in{\Z}}$ with the defining relations
\begin{equation}\label{e:1.1}
[a_n, a_m]=n\,\delta_{n+m,0} \, 1.
\end{equation}
The generators $a_n$ for $n\in {\Z}$ can be realized from the fermions $\phi_n, \phi_n^* (n\in{\Z})$,
which satisfy the following relations
\begin{equation}\label{e:1.2}
[\phi_n,\phi_m]_+=0, \,\, [\phi_n^*,\phi_m^*]=0,\,\, [\phi_n^*,\phi_m]_+=\delta_{n+m,0},
\end{equation}
where we use the notation $[X,Y]_+=XY+YX$. Let ${\mathcal A}$ be the Clifford algebra generated by the
$\phi_n,\phi^*_n, n\in{\Z}$. Then an element of ${\mathcal A}$ can be written as a finite linear
combination of monomials of the form
\[
\phi_{m_1}\cdots\phi_{m_r}\phi^*_{n_1}\cdots\phi^*_{n_s}, \quad {\text{where }} m_1<\cdots<m_r, \,\,
n_1<\cdots <n_s.
\]
By introducing a variable $k$, we define the fermionic generating functions as the formal
sums
\[
\phi(k)=\sum\limits_{n\in{\Z}+\frac{1}{2}} \phi_n k^{-n-\frac{1}{2}}, \quad
\phi^*(k)=\sum\limits_{n\in{\Z}+\frac{1}{2}}\phi^*_nk^{-n-\frac{1}{2}}.
\]
Then the $a_n$ can be realized by the following equation
\begin{equation}\label{e:fb}
\sum\limits_{n\in{\Z}}a_n k^{-n-1}=:\phi(k)\phi^*(k):
\end{equation}
where $::$ is the normal order, see \cite{MJD} for detail.

The integral form $H_{\Z}$ of Heisenberg algebra $H$ using the vertex operators in terms of $a_n$
is generated by $\{p_n,q_n\}_{n\in{\mathbb N}}$ with relations
\begin{eqnarray}
q_n p_m & = & \sum_{k\geq 0}p_{m-k} q_{n-k}, \label{e:1.4}\\
 q_n q_m & = & q_m q_n, \label{e:1.5}\\
 p_n p_m & = & p_m p_n, \label{e:1.6}
\end{eqnarray}
in which $p_{0} = q_{0} = 1$ and $p_{k} = q_{k} = 0$ for $k <0$,
thus the summation in Equation (\ref{e:1.4}) is finite.
These generators $q_n, \ p_m$ are obtained as the homogeneous
components in $z$ of the ¡°halves of vertex operators¡±
\begin{equation}\label{e:pqrelation}
\sum\limits_{m\ge 0} q_m z^{-m}=\exp\left(\sum_{m\geq 1} \frac{1}{m}a_{m} z^{-m}\right), \qquad \text{and} \qquad
\sum\limits_{m\ge 0} p_m z^m = \exp\left(\sum_{m\geq 1} \frac{1}{m}a_{-m} z^m\right),
\end{equation}
which play a very important role in the conformal field theory and the theory of quantum
algebras. We notice that $H_{\Z}\otimes {\Q}$ is also generated over $\Q$ by the sequence $\{a_n\}_{n\in\Z}$
with the relation (\ref{e:1.1}).

The authors use slightly different presentations of integral forms of Heisenberg algebra in literature.
Khovanov \cite{K} introduced a calculus of planar diagrams for biadjoint functors and degenerate affine
Hecke algebras, which led to an additive monoidal category whose Grothendieck ring contains some
integral form of the Heisenberg algebra (and Khovanov conjectured that they are isomorphic).

A deformed integral form $'H_{\Zt}$ of Heisenberg algebra is given as the algebra
generated over the ring $\Zt$ by $\{{\mathsf p}_n, {\mathsf q}_n\}_{n\in{\mathbb N}}$ with relations
\begin{eqnarray}
{\sq}_n {\sp}_m & = & \sum_{k\geq 0}[k+1]{\sp}_{m-k} {\sq}_{n-k}, \label{e:1.8}\\
 {\sq}_n {\sq}_m & = & {\sq}_m {\sq}_n, \label{e:1.9} \\
 {\sp}_n {\sp}_m & = & {\sp}_m {\sp}_n,\label{e:1.10}
\end{eqnarray}
where ${\sp}_{0} = {\sq}_{0} = 1$, \, ${\sp}_{k} = {\sq}_{k} = 0$ for $k <0$,
and the quantum integer $[k]=\frac{t^{k}-1}{t-1}$.
Under the relation (\ref{e:pqrelation})
between $\{{\sp}_n,{\sq}_n\}_{n\ge 0}$ and $\{{\sa}_n\}_{n\in \Z, n\ne 0}$ when the coefficients
are extended to ${\C}[t,t^{-1}]$, we can find that the deformed Heisenberg algebra $'H_{\Ct}$
is generated by $\{{\sa}_n\}_{n\in\Z}$ with the defining relation
\begin{equation}\label{e:arelation}
[{\sa}_n, {\sa}_m]=n(1+t^n)\delta_{n+m, 0}\, 1.
\end{equation}
The authors of the current paper give a categorified version of the
algebra $'H_{\Zt}$ in \cite{WWY} by the method of \cite{CL}.

In this short note, we will discuss some rings related to Khovanov's construction in \cite{K}
and the categorification of $'H_{\Zt}$ in \cite{WWY}, and show that Khovanov's construction
gives a categorification of the classical Heisenberg algebra $H_{\Z}$, a statement
conjectured by Khovanov \cite{K}.

The authors are grateful to Morningside Center of Chinese
Academy of Sciences for providing excellent research environment and financial
support to our seminar in mathematical physics.
This work is also partially supported by NSF grant 11031005 and grant KZ201210028032.

\section{Categorification of the deformed Heisenberg algebra $'H_{\Zt}$}\label{section:2}
Let $V=\C v$ be a 1-dimensional vector space over the field $\C$ of complex numbers,
$\Lambda^*(V)$ the exterior algebra of $V$ with the basis $B$ consisting of  $1, \, v$. We define a $\C$-linear map $tr:\ \Lambda^*(V)\rightarrow \C$ by setting
\[
tr (v)=1,\ \ \ tr(1)=0.
\]
Let $B^\vee$ be the basis of $\Lambda^*(V)$ dual to $B$ with respect to the non-degenerate bilinear form $\langle a,b \rangle := \tr(ab)$. Under this bilinear form, we see that $1^\vee=v,\ v^\vee=1$.
$\Lambda^*(V)$ is graded as $\Lambda^*(V)={\C}\oplus {\C}v$, with the degree of
$b\in \Lambda^*(V)$ denoted by $|b|$, and $|\alpha|=0$ for $0\ne \alpha\in {\C}$, $|v|=1$.

Let $'\H'$ be the $2$-category consisting of the following items.
\begin{itemize}
\item {\bf Objects ($0$-cells) of $'\H'$.} The objects are the integers.
\item {\bf $1$-morphisms ($1$-cells) of $'\H'$.} Let $\sP:\ n\longmapsto n+1$ and $\sQ:\ n+1\longmapsto n$ be
two fundamental $1$-morphisms, $\mathbf 1$ the identity $1$-morphism on the objects,
and the $1$-morphisms $\sP\lba l\rba$, $\sQ\lba l\rba$, and ${\mathbf 1}\lba l\rba$ for
integers $l$ which we call the shifted $1$-morphisms by dimension $l$. The $1$-morphisms
of $'\H'$ are generated from $\sP\lba l\rba, \sQ\lba l\rba$, and $\mathbf 1\lba l\rba$ by finite
compositions and finite direct sums. The dimension shifting on the $1$-morphisms satisfy
the equality $A\lba k\rba\cdot B\lba l\rba=A\cdot B\lba k+l\rba$ for $1$-morphisms $A$ and $B$.
Note that $\sP\lba 0\rba, \sQ\lba 0\rba, {\mathbf 1}\lba 0\rba$ and $\sP, \sQ, {\mathbf 1}$ are
respectively the same.
\item {\bf $2$-morphisms ($2$-cells) of $'\H'$.} The $2$-morphisms constitute a vector space over $\C$
which is generated by the
planar-oriented diagrams on the plane strip ${\R}\times [0,1]$
from the lower endpoints to the upper endpoints, modulo relative
to boundary isotopies and some local relations which are described in detail in the following
paragraphs. The endpoints in the planar string diagrams represent either $\sP$ or $\sQ$.
If the lower end or the top end doesn't have points, then it means that the
the $2$-morphism representing the planar string diagram comes from or goes to the $1$-morphism $\mathbf 1$.
The $2$-morphisms are also graded and the grading is compatible with the
grading of $1$-morphisms.
\end{itemize}

An upward oriented strand denotes the identity 2-morphism $\id: \sP \rightarrow \sP$ while a downward oriented strand denotes the identity 2-morphism $\id: \sQ \rightarrow \sQ$. Upward-oriented lines and downward-oriented lines carrying dots  labeled by elements $b \in \Lambda^*(V)$ are other 2-morphisms of $'\H'$, for example,
\begin{equation*}
\begin{tikzpicture}[>=stealth,baseline=25pt]
\draw (0,0) -- (0,1)[->];
\filldraw (0,0.5) circle (2pt)+(0.25,0) node {$b$};
\draw [shift={+(1.8,0.5)}](0,0) node{$\in {\rm Hom}_{'\H'}(\sP,\sP),$};
\draw [shift={+(4.5,0)}](0,0) -- (0,1)[<-];
\filldraw [shift={+(4.5,0.67)}](0,0) circle (2pt)+(0.25,0) node {$b'$};
\filldraw [shift={+(4.5,0)}](0,0.4) circle (2pt)+(0.25,0.4) node {$b''$};
\draw [shift={+(6.5,0.5)}](0,0) node{$\in {\rm Hom}_{'\H'}(\sQ,\sQ),$};
\draw [shift={+(9.2,0)}](0,0) -- (1,1)[<-];
\draw [shift={+(9.2,0)}](1,0) -- (0,1)[->];
\filldraw [shift={+(9.2,0)}](0.3,0.3) circle (2pt)+(0.45,0.4) node {$b$};
\draw [shift={+(11.9,0.5)}](0,0) node {$\in {\rm Hom}_{'\H'}(\sQ\sP,\sP\sQ).$};
\end{tikzpicture}
\end{equation*}
On planar diagrams, compositions of 2-morphisms are depicted upward from bottom to top.
All these are similar to \cite{CL}, so are the local relations which the planar diagrams are requested to satisfy.

The local relations are the following. First, the dots can move freely along strands and through intersections, for
example,
\begin{equation*}
\begin{split}
&\begin{tikzpicture}[>=stealth, baseline=25pt]
\draw [shift={+(0,0)}](0,0) -- (1,1)[<-];
\draw [shift={+(0,0)}](1,0) -- (0,1)[->];
\filldraw [shift={+(0,0)}](0.3,0.3) circle (2pt)+(0.45,0.4) node {$b$};
\draw [shift={+(1.85,0.5)}] node{ $=$};
\draw [shift={+(2.7,0)}](0,0) -- (1,1)[<-];
\draw [shift={+(2.7,0)}](1,0) -- (0,1)[->];
\filldraw [shift={+(2.7,0)}](0.8,0.8) circle (2pt)+(0.95,0.75) node {$b$};
\draw [shift={+(3.7,0)}] (0.3,0) node{,};
\draw [shift={+(6,0)}](0,0) -- (1,1)[<-];
\draw [shift={+(6,0)}](1,0) -- (0,1)[->];
\filldraw [shift={+(6,0)}](0.2,0.8) circle (2pt)+(0.35,0.9) node {$b$};
\draw [shift={+(7.85,0.5)}] node{ $=$};
\draw [shift={+(8.7,0)}](0,0) -- (1,1)[<-];
\draw [shift={+(8.7,0)}](1,0) -- (0,1)[->];
\filldraw [shift={+(8.7,0)}](0.8,0.2) circle (2pt)+(0.95,0.3) node {$b$};
\draw [shift={+(9.7,0)}] (0.3,0) node{,};
\end{tikzpicture}\\
&\begin{tikzpicture}[>=stealth]
\draw (0,0.75) arc (180:360:.5);
\draw (0,1.2) -- (0,0.75) ;
\draw (1,1.2) -- (1,0.75) [<-];
\filldraw  (0,0.9) circle (2pt)+(.25,0) node {$b$};
\draw [shift={+(2,0.6)}](0,0) node{$=$};
\draw [shift={+(3,0)}](0,0.75) arc (180:360:.5);
\draw [shift={+(3,0)}](0,1.2) -- (0,0.75) ;
\draw [shift={+(3,0)}](1,1.2) -- (1,0.75) [<-];
\filldraw [shift={+(3,0)}] (0.5,0.25) circle (2pt)+(.75,0.25) node {$b$};
\draw [shift={+(5,0.6)}](0,0) node{$=$};
\draw [shift={+(6,0)}](0,0.75) arc (180:360:.5);
\draw [shift={+(6,0)}](0,1.2) -- (0,0.75) ;
\draw [shift={+(6,0)}](1,1.2) -- (1,0.75) [<-];
\filldraw [shift={+(6,0)}] (1,0.9) circle (2pt)+(1.25,0.9) node {$b$};
\draw [shift={(6,0)}] (1.5, 0.4) node{.};
\end{tikzpicture}
\end{split}
\end{equation*}
%
%
%
%
Second, collision of dots is controlled by the multiplication in the exterior algebra ${\Lambda}^*(V)$:
\begin{equation*}
\begin{tikzpicture}[>=stealth]
\draw (0,0) -- (0,1)[->];
\filldraw (0,0.5) circle (2pt)+(-0.35,0) node {$b_1b_2$};
\draw [shift={+(0.8,0.5)}](0,0) node{$=$};
\draw [shift={+(1.6,0)}](0,0) -- (0,1)[->];
\filldraw [shift={+(1.6,0)}] (0,0.7) circle (2pt)+(0.25,0.7) node {$b_1$};
\filldraw [shift={+(1.6,0)}] (0,0.3) circle (2pt)+(0.25,0.3) node {$b_2$};
\draw [shift={+(2.6,0)}](0,0) node{$,$};

\draw [shift={+(5.6,0)}] (0,0) -- (0,1)[<-];
\filldraw [shift={+(5.6,0)}] (0,0.5) circle (2pt)+(-0.35,0.5) node {$b_1b_2$};
\draw [shift={+(6.4,0.5)}](0,0) node{$=$};
\draw [shift={+(7.2,0)}](0,0) -- (0,1)[<-];
\filldraw [shift={+(7.2,0)}] (0,0.7) circle (2pt)+(0.25,0.7) node {$b_1$};
\filldraw [shift={+(7.2,0)}] (0,0.3) circle (2pt)+(0.25,0.3) node {$b_2$};
\draw [shift={+(7.2,0)}](1,0) node{$.$};
\end{tikzpicture}
\end{equation*}
Third, dots on strands supercommute when they move past one another, for example:
\begin{equation*}
\begin{tikzpicture}[>=stealth]
\draw (0,0) -- (0,1)[->];
\filldraw (0,0.3) circle (2pt)+(-0.25,0) node {$b_1$};
\draw [shift={+(0.6,0.5)}](0,0) node{$\cdots$};
\draw [shift={+(1.2,0)}](0,0) -- (0,1)[->];
\filldraw [shift={+(1.2,0)}] (0,0.7) circle (2pt)+(0.25,0.7) node {$b_2$};
\draw [shift={+(2.2,0)}](0,0.5) node{$=$};
\draw [shift={+(3.7,0)}] (0,0.5) node {$(-1)^{|b_1||b_2|}$};
\draw [shift={+(5.2,0)}] (0,0) -- (0,1)[->];
\filldraw [shift={+(5.2,0)}] (0,0.7) circle (2pt)+(-0.25,0.7) node {$b_1$};
\draw [shift={+(5.8,0.5)}](0,0) node{$\cdots$};
\draw [shift={+(6.4,0)}](0,0) -- (0,1)[->];
\filldraw [shift={+(6.4,0)}] (0,0.3) circle (2pt)+(0.25,0.3) node {$b_2$.};
\end{tikzpicture}
\end{equation*}
Last, local relations contain the following:
\begin{align}\label{local1}
&\begin{tikzpicture}[>=stealth,baseline=0.8cm]
\draw (0,0) .. controls (0.8,0.8) .. (0,1.6)[->];
\draw (0.8,0) .. controls (0,0.8) .. (0.8,1.6)[->] ;
\draw (1.2,0.8) node {=};
\draw (2,0) --(2,1.6)[->];
\draw (2.8,0) -- (2.8,1.6)[->];
\draw [shift={+(2.8,0)}] (0.5,0) node{,};
\draw [shift={+(6.8,0)}](0,0) -- (1.6,1.6)[->];
\draw [shift={+(6.8,0)}](1.6,0) -- (0,1.6)[->];
\draw [shift={+(6.8,0)}](0.8,0) .. controls (0,0.8) .. (0.8,1.6)[->];
\draw [shift={+(6.8,0)}](2,0.8) node {=};
\draw [shift={+(6.8,0)}](2.4,0) -- (4,1.6)[->];
\draw [shift={+(6.8,0)}](4,0) node{\qquad ;} -- (2.4,1.6)[->];
\draw [shift={+(6.8,0)}](3.2,0) .. controls (4,0.8) .. (3.2,1.6)[->];
\end{tikzpicture}\\
\label{local2}
&\begin{tikzpicture}[>=stealth,baseline=0.8cm]
\draw (0,0) .. controls (0.8,0.8) .. (0,1.6)[<-];
\draw (0.8,0) .. controls (0,0.8) .. (0.8,1.6)[->];
\draw (1.2,0.8) node {=};
\draw (1.6,0) --(1.6,1.6)[<-];
\draw (2.4,0) -- (2.4,1.6)[->];
\draw (2.8,0.8) node {$-$};
\draw (3.2,1.4) arc (180:360:.4);
\draw (3.2,1.6) -- (3.2,1.4) ;
\draw (4,1.6) -- (4,1.4) [<-];
\draw (4,.2) arc (0:180:.4) ;
\filldraw  (3.6,1) circle (2pt)+(.2,0) node {$v$};
\draw (4,0) -- (4,.2) ;
\draw (3.2,0) -- (3.2,.2) [<-];
\draw (4.4,0.8) node {$-$};
\filldraw  (5.2,0.6) circle (2pt)+(.2,0) node {$v$};
\draw (4.8,1.4) arc (180:360:.4);
\draw (4.8,1.6) -- (4.8,1.4) ;
\draw (5.6,1.6) -- (5.6,1.4) [<-];
\draw (5.6,.2) arc (0:180:.4);
\draw (5.6,0) -- (5.6,.2) node{\qquad ;} ;
\draw (4.8,0) -- (4.8,.2) [<-];
\end{tikzpicture}\\
\label{local3}
&\begin{tikzpicture}[>=stealth,baseline=0.8cm]
\draw (0,0) .. controls (0.8,0.8) .. (0,1.6)[->];
\draw (0.8,0) .. controls (0,0.8) .. (0.8,1.6)[<-] ;
\draw (1.2,0.8) node {=};
\draw (1.84,0) --(1.84,1.6)[->];
\draw (2.64,0) -- (2.64,1.6)[<-];
\draw [shift={+(2.64,0)}] (0.5,0) node{,};
\draw [shift={+(7.5,0.8)}](-0.8,0) .. controls (-0.8,.4) and (-.24,.4) .. (-.08,0) ;
\draw [shift={+(7.5,0.8)}](-0.8,0) .. controls (-0.8,-.4) and (-.24,-.4) .. (-.08,0) ;
\draw [shift={+(7.5,0.8)}](0,-0.8) .. controls (0,-.4) .. (-.08,0) ;
\draw [shift={+(7.5,0.8)}](-.08,0) .. controls (0,.4) .. (0,0.8) [->] ;
\draw [shift={+(7.5,0.8)}](0.56,0) node {$=0$\quad ;};
\end{tikzpicture}\\
\label{local4}
&\begin{tikzpicture}[>=stealth,baseline=0cm]
\draw [shift={+(0,0)}](0,0) arc (180:360:0.5cm) ;
\draw [shift={+(0,0)}][->](1,0) arc (0:180:0.5cm) ;
\filldraw [shift={+(1,0)}](0,0) circle (2pt);
\draw [shift={+(0,0)}](1,0) node [anchor=east] {$b$};
\draw [shift={+(0,0)}](1.75,0) node{$\qquad = \tr(b)$ \quad .};
\end{tikzpicture}
\end{align}
We notice that the above local relations imply that the second relation of (\ref{local1}) is
satisfied when any of the arrows of the strings in the diagram on the
left hand side of the equality is reversed (and the arrows of the strings of the right
hand side of the equality are reversed correspondingly), hence we
depict this local relation as:
\begin{equation}\label{local1a}
\begin{tikzpicture}[>=stealth,baseline=0.8cm]
  \draw [shift={+(4.8,0)}](0,0) -- (1.6,1.6)[-];
  \draw [shift={+(4.8,0)}](1.6,0) -- (0,1.6)[-];
  \draw [shift={+(4.8,0)}](0.8,0) .. controls (0,0.8) .. (0.8,1.6)[-];
  \draw [shift={+(5,0)}](2,0.8) node {=};
  \draw [shift={+(5.2,0)}](2.4,0) -- (4,1.6)[-];
  \draw [shift={+(5.2,0)}](4,0) node{\qquad ,} -- (2.4,1.6)[-];
  \draw [shift={+(5.2,0)}](3.2,0) .. controls (4,0.8) .. (3.2,1.6)[-];
\end{tikzpicture}
\end{equation}
and first relation of (\ref{local1}) is satisfied when both of the arrows of the strings
in the diagram are reversed.
We define the degrees of the planar string diagrams (as morphisms from
$X_1X_2\cdots X_k\to Y_1Y_2\cdots Y_l$ with $k\ge 0, l\ge 0$ integers,
$X_i,Y_j\in\{\sP,\sQ\}$, and $X_1X_2\cdots X_k={\mathbf 1}$ if $k=0$) from the degrees of
the simplest planar string diagrams:
\[
\begin{split}
&\begin{tikzpicture}[>=stealth]
\draw  (-.5,.5) node {$deg$};
\draw [->](0,0) -- (1,1);
\draw [->](1,0) -- (0,1);
\draw [shift={+(2,0)}] (-0.5,0.5) node {$=$};
\draw [shift={+(2.5,0)}] (0,.5) node {$deg$};
\draw [shift={+(3.0,0)}][<-](0,0) -- (1,1);
\draw [shift={+(3.0,0)}][<-](1,0) -- (0,1);
\draw [shift={+(5,0)}] (-0.5,0.5) node {$=$};
\draw [shift={+(5.5,0)}] (0,.5) node {$deg$};
\draw [shift={+(6.0,0)}][->](0,0) -- (1,1);
\draw [shift={+(6.0,0)}][<-](1,0) -- (0,1);
\draw [shift={+(8,0)}] (-0.5,0.5) node {$=$};
\draw [shift={+(8.5,0)}] (0,.5) node {$deg$};
\draw [shift={+(9.0,0)}][<-](0,0) -- (1,1);
\draw [shift={+(9.0,0)}][->](1,0) -- (0,1);
\draw [shift={+(9.0,0)}](1.5,.5) node{$ \quad = 0 \quad ;$};
\end{tikzpicture}\\
&\begin{tikzpicture}[>=stealth]
\draw  (-.5,-.25) node {$deg$};
\draw (0,0) arc (180:360:.5)[->] ;
\draw (1.75,-.25) node{$ \qquad =-1 \,\, ,\quad $};
\draw (3,-.25) node{$deg$};
\draw (4.5,-.5) arc (0:180:.5) [->];
\draw (5,-.25) node{$ \quad =0 \quad ;$};
\end{tikzpicture}\\
&\begin{tikzpicture}[>=stealth]
\draw  (-.5,-.25) node {$deg$};
\draw (0,0) arc (180:360:.5)[<-] ;
\draw (1.75,-.25) node{$ \quad = 0 \quad , \quad $};
\draw (3,-.25) node{$deg$};
\draw (4.5,-.5) arc (0:180:.5) [<-];
\draw (5,-.25) node{$ \quad = 1 \quad ; $};
\end{tikzpicture}
\end{split}
\]
and the degree of a dot labeled by $b$ equals the degree of $b$ in the graded algebra $\Lambda^*(V)$.
If $f$ and $g$ are two planar strings, and they can be composed as $g\circ f\ne 0$, then we
define $deg(g\circ f):= deg(f) + deg(g)$. The degree of a string diagram is defined to be
the total sum of the degrees of all the planar strings in the string diagram. And finally,
if $f=\sum_i c_i f_i$ with each $f_i\ne 0$ a string diagram and $0\ne c_i\in\C$, then we define
\[
deg(f)=\max_i\{deg(f_i)\}.
\]
and the degree of a dot labeled by $b$ equals the degree of $b$ in the graded algebra $\Lambda^*(V)$. Thus for any two 1-morphisms $A,\ B$, the vector
space ${\rm Hom}_{'\H'}(A,B)$ of 2-morphisms
$A\rightarrow B$ is graded, and the composition of 2-morphisms is compatible with grading.

The $2$-morphisms of $'\H'$ constitute a ring with the obvious addition and the multiplication given
by the composition of planar string diagrams. We define
\[
f\cdot g=0
\]
if the two planar strings $f$ and $g$ can not be composed.

For an $f\in M={\rm Hom}_{'\H'}(X\lba l_1\rba,Y\lba l_2\rba)$, where
$X=X_1X_2\cdots X_{k_1}$ and $Y=Y_1Y_2\cdots Y_{k_2}$ with $X_i, Y_j\in \{\sP,\sQ\}$
(that is, $f$ is a planar string which connects the sequence of $X_i$'s on the
bottom and the sequence of $X_i$'s on the top), we also
view $f\in N={\rm Hom}_{'\H'}(X, Y)$. We then define
\[
deg_{M}(f)=deg_{N}(f)+l_1-l_2=deg(f)+l_1-l_2.
\]

In addition, we define the shifted degree (which is denoted by $sdeg$) of the planar diagrams
(also viewed as $2$-morphisms from $X_1X_2\cdots X_k\to Y_1Y_2\cdots Y_l$ with
$X_i,Y_j\in \{\sP,\sQ\}$):
\[
\begin{split}
&\begin{tikzpicture}[>=stealth]
\draw  (-.5,.5) node {$sdeg$};
\draw [->](0,0) -- (1,1);
\draw [->](1,0) -- (0,1);
\draw [shift={+(2,0)}] (-0.5,0.5) node {$=$};
\draw [shift={+(2.5,0)}] (0,.5) node {$sdeg$};
\draw [shift={+(3.0,0)}][<-](0,0) -- (1,1);
\draw [shift={+(3.0,0)}][<-](1,0) -- (0,1);
\draw [shift={+(5,0)}] (-0.5,0.5) node {$=$};
\draw [shift={+(5.5,0)}] (0,.5) node {$sdeg$};
\draw [shift={+(6.0,0)}][->](0,0) -- (1,1);
\draw [shift={+(6.0,0)}][<-](1,0) -- (0,1);
\draw [shift={+(8,0)}] (-0.5,0.5) node {$=$};
\draw [shift={+(8.5,0)}] (0,.5) node {$sdeg$};
\draw [shift={+(9.0,0)}][<-](0,0) -- (1,1);
\draw [shift={+(9.0,0)}][->](1,0) -- (0,1);
\draw [shift={+(9.0,0)}](1.5,.5) node{$ \quad = 0 \quad ;$};
\end{tikzpicture}\\
&\begin{tikzpicture}[>=stealth]
\draw  (-.5,-.25) node {$sdeg$};
\draw (0,0) arc (180:360:.5)[->] ;
\draw (1.75,-.25) node{$ \qquad =-1 \,\, ,\quad $};
\draw (3,-.25) node{$sdeg$};
\draw (4.5,-.5) arc (0:180:.5) [->];
\draw (5,-.25) node{$ \quad =0 \quad ;$};
\end{tikzpicture}\\
&\begin{tikzpicture}[>=stealth]
\draw  (-.5,-.25) node {$sdeg$};
\draw (0,0) arc (180:360:.5)[<-] ;
\draw (1.75,-.25) node{$ \quad = 0 \quad , \quad $};
\draw (3,-.25) node{$sdeg$};
\draw (4.5,-.5) arc (0:180:.5) [<-];
\draw (5,-.25) node{$ \quad = 1 \quad ; $};
\end{tikzpicture}
\end{split}
\]
and the shifted degree $sdeg(b)$ of a dot labeled by $b \in \Lambda^*(V)$ is
the same as $deg(b)$ defined above. By convention, we set $deg(0)=sdeg(0)=+\infty$.

The local relations (\ref{local2}), (\ref{local3}), and (\ref{local4}) imply that
$\sQ\sP\cong \sP\sQ\oplus{\mathbf 1}\oplus {\mathbf 1}\lba 1\rba$, which is illustrated
in the category $'\H'$ as:
\begin{equation}\label{e:qpdecomp}
\xy
(0,20)*+ {\sQ\sP}="T1";
(-25,0)*+ {\sP\sQ}="m1";
(0,0)*+ {\mathbf 1}="m2";
(25,0)*+{{\mathbf 1}\lba 1\rba \,\, .}="m3";
(0, -20)*+{\sQ\sP}="T2";
 {\ar@/_0pc/_{ \begin{tikzpicture} \draw [<-] (0,0) -- (0.4,0.4);
 \draw [->] (0.4,0) -- (0,0.4);\end{tikzpicture}
  \xy
 \endxy}
  "T1";"m1"};
 {\ar@/_0.0pc/_{ \begin{tikzpicture} \draw (0.3,-.5) arc (-15:195:.25) [->];\end{tikzpicture} \xy
  \endxy} "T1";"m2"};
 {\ar@/^0pc/^{ \begin{tikzpicture} \draw (0.3,-.5) arc (-15:195:.25) [->];
  \filldraw (0.175,-0.225) circle (1.5pt); \draw (0.29,-0.1) node{$v$}; \end{tikzpicture}
  \xy
  \endxy} "T1";"m3"};
 {\ar@/_0pc/_{  \begin{tikzpicture} \draw [->] (0,0) -- (0.4,0.4);
 \draw [<-] (0.4,0) -- (0,0.4);\end{tikzpicture}
  \xy
  \endxy} "m1";"T2"};
 {\ar@/_0.0pc/_{ \begin{tikzpicture} \draw [shift={+(0,0.3)}](0,0) arc (180:360:.25)[->];
  \filldraw (0.125,0.105) circle (1.5pt) +(0,0.25) node{$v$}; \end{tikzpicture}
 \xy
    \endxy} "m2";"T2"};
 {\ar@/^0pc/^{ \begin{tikzpicture} \draw [shift={+(0,0.3)}](0,0) arc (180:360:.25)[->];\end{tikzpicture}
  \xy
    \endxy} "m3";"T2"};
\endxy
\end{equation}

By (\ref{local1}), we see that upward oriented crossings satisfy the relations of symmetric group
$S_n$. We have a canonical homomorphism
\[
{\C}[S_n] \longrightarrow {\End}_{'\H'}(\sP^n)
\]
from the group algebra of $S_n$ to the endomorphism ring of the $n$-th tensor power of $P$.
Through some diagram manipulation according to $2$-morphisms of $'\H'$, we can see that the
relation (\ref{local1}) is also satisfied by the downward oriented crossings, therefore
get the similar canonical homomorphism
\[
{\C}[S_n] \longrightarrow {\End}_{'\H'}(\sQ^n).
\]
For a positive integer $n$ and a partition $\lambda$ of $n$, let $e_{\lambda}$ be the associated Young
symmetrizer of ${\C}[S_n]$. The Young symmetrizers are idempotents of the group algebra ${\C}[S_n]$ given
by
\[
e_{\lambda}=a_{\lambda} b_{\lambda}/{n_{\lambda}}
\]
where $n_{\lambda}={n!/{\rm dim}(V_{\lambda})}$ with $V_{\lambda}$ the irreducible representation
corresponding to the partition $\lambda$, the elements
$a_{\lambda}, b_{\lambda}\in {\C}[S_n]$ are defined for a Young tableau $T$ of the partition
$\lambda$ by
\[
a_{\lambda}=\sum\limits_{g\in L_{\lambda}} g, \qquad b_{\lambda}=\sum\limits_{g\in L'_{\lambda}} sign(g)g,
\]
and
\[
\begin{split}
&L_{\lambda}=\{g\in S_n: \, g\text{ preserves each row of $T$}\}, \\
&L'_{\lambda}=\{g\in S_n: \, g\text{ preserves each column of $T$}\}.
\end{split}
\]
The irreducible representations of $S_n$ are given by ${\C}[S_n]e_{\lambda}$ for partitions $\lambda$ of
$n$. Let $e_{(n)}$ be the idempotent corresponding to the partition $(n)$ of $n$, with
${\C}[S_n]e_{(n)}$ being the trivial representation of $S_n$. See \cite{FH} for detail.

Let $'\H$ be the Karoubi envelope of $'\H'$. This is a $2$-category described in the following.
\begin{itemize}
\item {\bf Objects of $'\H$:} same as the objects of $'\H'$.
\item {\bf $1$-morphisms of $'\H$:} pairs $(M,e)$, where $M$ is a $1$-morphism of $'\H'$, and
$e$ is an idempotent $2$-morphism $M\xrightarrow{e} M$ of $'\H'$.
\item {\bf $2$-morphisms of $'\H$:} given two $1$-morphisms $(M, e)$ and $(M',e')$, the set
${\rm Hom}((M,e), (M',e'))$ consists of the $2$-morphisms $f:M\to M'$ of $'\H'$ such that
the diagram
\[
\xymatrix{M\ar[r]^{f}\ar[d]_{e} & M'\ar[d]^{e'}\\
M\ar[r]_{f} & M'}.
\]
\end{itemize}

The $2$-morphisms of $'\H$ constitute a ring, as so do the $2$-morphisms of $'\H'$.

For convenience, we also view $'\H'$ and $'\H$ as categories with the $1$-morphisms considered
as the objects and the $2$-morphisms as the morphisms (hence forgetting the set $\Z$ of $0$-cells).
Then $'\H'$ is a subcategory of $'\H$, and both are $\C$-linear additive monoidal categories.
For an idempotent $e:M\to M$ in $'\H'$ with $M$ a $1$-morphism, ${\id}_{M}-e$ is also an
idempotent, and we have a direct sum decomposition in $'\H$:
\[
M=(M,e)\oplus (M,{\id}_{M}-e).
\]

The 1-morphisms $\sP^n$ and $\sQ^n$
splits into direct sums of $(\sP^n,e_{\lambda})$ and $(\sQ^n,e_{\lambda})$ in $'\H$ respectively, over
partitions $\lambda$ of $n$. Define
\[
\sP_{\lambda}=(\sP^n,e_{\lambda}), \qquad \sQ_{\lambda}=(\sQ^n,e_{\lambda}), \qquad
\sP_n = (\sP^n, e_{(n)}), \qquad  \sQ_n=(\sQ^n, e_{(n)}).
\]

\begin{theorem}[\cite{CL}, \cite{WWY}]\label{theorem:1}
In the $2$-category $'\mathcal{H}$, the $1$-morphisms $\sP_n$ and $\sQ_n$ satisfy
\begin{alignat}{5}
\sQ_n \sP_m  &\cong &&\,\,\bigoplus_{k\geq 0}\bigoplus_{l=0}^k \sP_{m-k}\sQ_{n-k}\lba l\rba,\label{$1}\\
\sQ_n \sQ_m  &\cong &&\,\,\sQ_m \sQ_n,\label{$2}\\
\sP_n \sP_m  &\cong &&\,\,\sP_m \sP_n.\label{$3}
\end{alignat}
\end{theorem}

From Theorem \ref{theorem:1}, we get a natural homomorphism of algebras
\begin{equation}\label{e:pi}
\pi:\,\, 'H_{\Zt} \longrightarrow K_0('\H)
\end{equation}
by sending $p_m$ to $[\sP_m]$, $q_n$ to $[\sQ_n]$, and $t^l$ to $[{\mathbf 1}\lba l\rba]$, where
$[X]$ denotes the isomorphism class of the $1$-morphism $X$ of $'\H$.

\begin{theorem}[\cite{WWY}]\label{theorem:1.5}
$\pi$ is an isomorphism of algebras.
\end{theorem}

\section{Khovanov's Construction of Heinsenberg categorification}\label{section:3}
In this section, we introduce Khovanov's construction \cite{K} of Heisenberg categorification $\H$
of the integral form $H_{\Z}$ of the classical Heisenberg algebra. Recall that $H_{\Z}$ is
the associative algebra generated by $p_n$, $q_n$ for $n\ge 1$ with the relations (\ref{e:1.4})--(\ref{e:1.6}),
hence the sequence $\{a_n\}_{n\in\Z}$ satisfies Equation (\ref{e:1.1}).

Let
\[
\sum_{m\ge 0} {\tp}_m z^m=\exp\left(\sum_{m\ge 1}(-1)^{m-1}\frac{a_{-m}}{m} z^m\right).
\]
Then ${\tp}_0=1$ and
\[
\left(1+\sum_{m\ge 1}{\tp}_m z^m\right)\cdot \left(1+\sum_{m\ge 1}(-1)^m p_m z^m\right)=1.
\]
The symbols $\tp_n, q_n$, $n\ge 1$ satisfy
\begin{align}
& q_n\tp_m = \tp_m q_n + \tp_{m-1} q_{n-1}, \tag{\ref{e:1.4}$'$}\label{e:1.4pr}\\
& \tp_n\tp_m = \tp_m \tp_n. \tag{\ref{e:1.6}$'$}\label{e:1.6pr}
\end{align}

We have a $2$-category $\H'$ consisting of
\begin{itemize}
\item {\bf objects:} all rational integers;
\item {\bf $1$-morphisms:} the symbols generated by $\mathcal P$ and $\mathcal Q$ through finite
compositions and finite direct sums, where ${\mathcal P}: n\longmapsto n+1$ and ${\mathcal Q}:
n+1\longmapsto n$ are two chosen $1$-morphisms;
\item {\bf $2$-morphisms:} a $\C$-vector space generated by the planar string diagrams
in the strip ${\R}\times [0,1]$ modulo relative to boundary isotopies and local relations
(\ref{local1}), (\ref{local2.1}), (\ref{local3}), and (\ref{local4.1}),
\end{itemize}
where the local relations (\ref{local2.1}) and (\ref{local4.1}) are:
\begin{align}\label{local2.1}\tag{\ref{local2}$'$}
&\begin{tikzpicture}[>=stealth,baseline=0.8cm]
\draw (0,0) .. controls (0.8,0.8) .. (0,1.6)[<-];
\draw (0.8,0) .. controls (0,0.8) .. (0.8,1.6)[->];
\draw [shift={+(0.5,0)}](1.2,0.8) node {=};
\draw [shift={+(1,0)}](1.6,0) --(1.6,1.6)[<-];
\draw [shift={+(1,0)}](2.4,0) -- (2.4,1.6)[->];
\draw [shift={+(1.5,0)}](2.8,0.8) node {$-$};
\draw [shift={+(2,0)}](3.2,1.4) arc (180:360:.4);
\draw [shift={+(2,0)}](3.2,1.6) -- (3.2,1.4) ;
\draw [shift={+(2,0)}](4,1.6) -- (4,1.4) [<-];
\draw [shift={+(2,0)}](4,.2) arc (0:180:.4);
\draw [shift={+(2,0)}](4,0) -- (4,.2) ;
\draw [shift={+(2,0)}](3.2,0) -- (3.2,.2) [<-];
\draw [shift={+(2.1,0)}](4.1,0) node{;};
\end{tikzpicture}\\
\label{local4.1}\tag{\ref{local4}$'$}
&\begin{tikzpicture}[>=stealth,baseline=0cm]
\draw [shift={+(0,0)}](0,0) arc (180:360:0.5cm) ;
\draw [shift={+(0,0)}][->](1,0) arc (0:180:0.5cm) ;
\draw [shift={+(0.2,0)}](1.75,0) node{$=$};
\draw [shift={+(1,0)}] (2,0) node{$1\,\,\,.$};
\end{tikzpicture}
\end{align}
The $2$-morphisms in $\H'$ constitute of a ring with the same addition and multiplication
rules as in $'\H'$ except the local relations stated as above. The local relations
(\ref{local2.1}), (\ref{local3}), and (\ref{local4.1}) of $\H'$ imply
$\mcQ\mcP\cong \mcP\mcQ\oplus {\mathbf 1}$, which is illustrated in the category $\H'$ as:
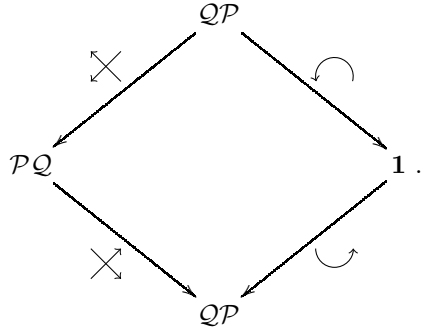
\begin{equation*}
\xy
(0,20)*+ {\mcQ\mcP}="T1";
(-25,0)*+ {\mcP\mcQ}="m1";
(25,0)*+{{\mathbf 1} \,\, .}="m3";
(0, -20)*+{\mcQ\mcP}="T2";
 {\ar@/_0pc/_{ \begin{tikzpicture} \draw [<-] (0,0) -- (0.4,0.4);
 \draw [->] (0.4,0) -- (0,0.4);\end{tikzpicture}
  \xy
 \endxy}
  "T1";"m1"};
 {\ar@/^0pc/^{ \begin{tikzpicture} \draw (0.3,-.5) arc (-15:195:.25) [->]; \end{tikzpicture}
  \xy
  \endxy} "T1";"m3"};
 {\ar@/_0pc/_{  \begin{tikzpicture} \draw [->] (0,0) -- (0.4,0.4);
 \draw [<-] (0.4,0) -- (0,0.4);\end{tikzpicture}
  \xy
  \endxy} "m1";"T2"};
 {\ar@/^0pc/^{ \begin{tikzpicture} \draw [shift={+(0,0.3)}](0,0) arc (180:360:.25)[->];\end{tikzpicture}
  \xy
    \endxy} "m3";"T2"};
\endxy
\end{equation*}

Let $\H$ be the Karoubi envelope of $\H'$.
It is clear that ${\C}[S_n]\subset {\End}(\mcP^n)$, and ${\C}[S_n]\subset {\End}(\mcQ^n)$,
hence we define the following $1$-morphisms in $\H$ from Young symmetrizers $e_{\lambda}$
\[
{\mcP}_{\lambda}=(\mcP^n, e_{\lambda}),\qquad {\mcQ}_{\lambda}=(\mcQ^n, e_{\lambda}), \qquad
{\mcP}_n=(\mcP^n, e_{(n)}), \qquad {\mcQ}_n=(\mcQ^n, e_{(n)}), \qquad
{\mcPb}_n=(\mcPb^n, e_{(1^n)}),
\]
where $\lambda$ denotes a partition of the positive integer $n$ and
$e_{(1^n)}$ denotes the  anti-symmetrizer of $S_n$. Then we have
\begin{theorem}[\cite{K},\cite{CL},\cite{WWY}]\label{theorem:3}
In the $2$-category $\H$, there are isomorphisms between $1$-morphisms:
\begin{alignat}{3}
{\mcQ}_n \, {\mcPb}_m & \cong  {\mcPb}_m\, {\mcQ}_n \oplus {\mcPb}_{m-1}\, {\mcQ}_{n-1}\, ,
\label{e:21}\\
{\mcQ}_n \, {\mcP}_m & \cong  {\mcP}_m\, {\mcQ}_n \oplus {\mcP}_{m-1}\, {\mcQ}_{n-1}
\oplus\mcP_{m-2}\, \mcQ_{n-2}\oplus \cdots\,,
\label{e:22}\\
{\mcQ}_n \, {\mcQ}_m & \cong  {\mcQ}_m\, {\mcQ}_n\,,
\label{e:23}\\
{\mcPb}_n\, \mcPb_m & \cong  \mcPb_m \, \mcPb_n\,, \label{e:24}\\
{\mcP}_n\, \mcP_m & \cong  \mcP_m \, \mcP_n\,. \label{e:25}
\end{alignat}
\end{theorem}
\begin{proof}
Equations (\ref{e:21}), (\ref{e:23}), and (\ref{e:24}) are proved in \cite{K}. They correspond
to Equations (\ref{e:1.4pr}), (\ref{e:1.5}), and (\ref{e:1.6pr}). Equation (\ref{e:22}) and (\ref{e:25}) can
be proved in the same way as explained in \cite{CL} or \cite{WWY}.
\end{proof}

Due to the relations in Theorem \ref{theorem:3}, there exists a ring homomorphism
\begin{equation}\label{e:gamma}
\gamma: H_{\Z} \lra K_0(\H)
\end{equation}
such that
\begin{equation*}
\gamma(p_n)=[\mcP_n], \qquad \gamma(q_n)=[\mcQ_n]
\end{equation*}
which imply that $\gamma(\tp_n)=[\mcPb_n]$.
\begin{theorem}[\cite{K}]\label{theorem:4}
The ring homomorphism $\gamma$ is injective.
\end{theorem}
It is also conjectured in \cite{K} that $\gamma$ is a ring isomorphism.

The Karoubi envelope
$\H$ of $\H'$ is determined by the pairs $(M,e)$ for $M$ a $1$-morphism of $\H'$ and
$e\in \EndHp(M)$ an idempotent.
Let $\tilde{\mathcal R}_{\H'}$ be the $\C$ vector space generated by
the planar string diagrams on the strip ${\R}\times [0,1]$ from the lower endpoints
to the upper endpoints, modulo relative to boundary isotopies. Then $\tilde{\mcR}_{\H'}$ is a ring with
the addition and the multiplication mentioned above. It is generated by the following
diagrams
\begin{equation}\label{e:26}
\begin{split}
&\id_{\mcP}=\begin{tikzpicture}[baseline=0.25cm]\draw[->] (0,0) -- (0,0.7);\end{tikzpicture}\,\,\,\,,\qquad
\id_{\mcQ}=\begin{tikzpicture}[baseline=0.25cm]\draw[<-] (0,0) -- (0,0.7);\end{tikzpicture}\,\,\,\,,\qquad
\varpi_1=\begin{tikzpicture}[baseline=0.25cm]\draw [<-](0,0) -- (0.7,0.7);
\draw [->](0.7,0) -- (0,0.7);\end{tikzpicture}\,\,\,,\qquad
\varphi_1=\begin{tikzpicture}[baseline=0.25cm]\draw [->](0,0) -- (0.7,0.7);
\draw [<-](0.7,0) -- (0,0.7);\end{tikzpicture}\,\,\,,\qquad
\vartheta_1=\begin{tikzpicture}[baseline=0.25cm]\draw [->](0,0) -- (0.7,0.7);
\draw [->](0.7,0) -- (0,0.7);\end{tikzpicture}\,\,\,,\qquad
\vartheta_2=\begin{tikzpicture}[baseline=0.25cm]\draw [<-](0,0) -- (0.7,0.7);
\draw [<-](0.7,0) -- (0,0.7);\end{tikzpicture}\,\,\,,\qquad \\
&\varpi_2=\begin{tikzpicture}[baseline=0.5cm]\draw (4.7,.35) arc (0:180:.5) [->]; \end{tikzpicture}\,\,\,,\qquad
\varphi_2=\begin{tikzpicture}[baseline=0.5cm]\draw (7.5,.85) arc (180:360:.5) [->];
\end{tikzpicture}\,\,\,,\qquad
\vartheta_3=\begin{tikzpicture}[baseline=0.5cm]\draw (4.7,.35) arc (0:180:.5) [<-]; \end{tikzpicture}\,\,\,,\qquad
\vartheta_4=\begin{tikzpicture}[baseline=0.5cm]\draw (7.5,.85) arc (180:360:.5) [<-];
\end{tikzpicture}\,\,\,.\qquad
\end{split}
\end{equation}
through compositions and $\C$-linear extensions on $1$-morphisms of $\H'$, which we mean that,
for any $1$-morphisms $X=X_1X_2\cdots X_k$ and $Y=Y_1Y_2\cdots Y_l$ with
$X_i, Y_j\in \{{\mathbf 1},\mcP,\mcQ\}$, the $\C$-vector space $\HomHp(X,Y)$ is generated by
the compositions of the $2$-morphisms  $\alpha_1\,\,\alpha_2\,\,\cdots\,\,\alpha_m: X\lra Y$,
where each $\alpha_i$ is one of the string diagrams in (\ref{e:26}) and
$\alpha_1\,\,\alpha_2\,\,\cdots\,\,\alpha_m$ denotes the horizontal composition of $2$-morphisms,
for example:
\[
\begin{tikzpicture}
\draw (0,0) -- (8,0);
\filldraw (1,0) circle (0.5pt);
\draw [shift={+(0.1,-0.3)}] (1,0) node{$\mcP$};
\filldraw (2,0) circle (0.5pt);
\draw [shift={+(0.1,-0.3)}] (2,0) node{$\mcQ$};
\draw [->] (1,0) arc (180:0:0.5);
\filldraw (3,0) circle (0.5pt);
\draw [shift={+(0.1,-0.3)}] (3,0) node{$\mcP$};
\filldraw (4,0) circle (0.5pt);
\draw [shift={+(0.1,-0.3)}] (4,0) node{$\mcP$};
\filldraw (5,0) circle (0.5pt);
\draw [shift={+(0.1,-0.3)}] (5,0) node{$\mcQ$};

\draw (0,2) -- (8,2);
\filldraw (3,2) circle (0.5pt);
\draw [shift={+(0.1,0.3)}] (3,2) node{$\mcP$};
\filldraw (4,2) circle (0.5pt);
\draw [shift={+(0.1,0.3)}] (4,2) node{$\mcP$};
\filldraw (5,2) circle (0.5pt);
\draw [shift={+(0.1,0.3)}] (5,2) node{$\mcQ$};
\filldraw (6,2) circle (0.5pt);
\draw [shift={+(0.1,0.3)}] (6,2) node{$\mcQ$};
\filldraw (7,2) circle (0.5pt);
\draw [shift={+(0.1,0.3)}] (7,2) node{$\mcP$};
\draw [->] (6,2) arc (180:360:0.5);

\draw [->] (3,0) -- (4,2);
\draw [->] (4,0) -- (3,2);
\draw [<-] (5,0) -- (5,2);
\end{tikzpicture}
\]
Let ${\sim}^{\H'}$ be the local relations (\ref{local1}), (\ref{local2.1}), (\ref{local3}), and
(\ref{local4.1}), then ${\mathcal R}_{\H'}:= \tilde{\mathcal R}_{\H'}/{\sim}^{\H'}$ is the ring of
$2$-morphisms of $\H'$. We now define a sub-category $\mathcal G'$ of $\H'$:
\begin{itemize}
\item {\bf objects:} all rational integers;
\item {\bf $1$-morphisms:} the $1$-morphisms of $\H'$ generated from
$\mathbf 1$, $P=\mcP$ and $Q=\mcQ\oplus \mcQ$;
\item {\bf $2$-morphisms:} the $\C$-vector space generated by
the following replacements of those diagrams in (\ref{e:26}):
\begin{equation}\label{e:29}
\begin{array}{rlll}
&\begin{tikzpicture}[baseline=0.25cm]\draw[->] (0,0) -- (0,0.7);\end{tikzpicture}=\id_{P}\,\,\,\,,\qquad
\begin{tikzpicture}[baseline=0.25cm]\draw[<-] (0,0) -- (0,0.7);\end{tikzpicture}=\id_{Q}\,\,\,\,,\qquad
&\begin{tikzpicture}[baseline=0.25cm]\draw [<-](0,0) -- (0.7,0.7);
\draw [->](0.7,0) -- (0,0.7);\end{tikzpicture}
=\left(\begin{array}{cc}\varpi_1 &0\\0 &\varpi_1\end{array}\right)\,\,\,,\qquad
&\begin{tikzpicture}[baseline=0.25cm]\draw [->](0,0) -- (0.7,0.7);
\draw [<-](0.7,0) -- (0,0.7);\end{tikzpicture}
=\left(\begin{array}{cc}\varphi_1 &0\\0 &\varphi_1\end{array}\right)\,\,\,,\qquad \\
&\begin{tikzpicture}[baseline=0.25cm]\draw [->](0,0) -- (0.7,0.7);
\draw [->](0.7,0) -- (0,0.7);\end{tikzpicture}=\vartheta_1\,\,\,,\qquad
&\begin{tikzpicture}[baseline=0.25cm]\draw [<-](0,0) -- (0.7,0.7);
\draw [<-](0.7,0) -- (0,0.7);\end{tikzpicture}
=\left(\begin{array}{cc}\vartheta_2 & 0\\0 & \vartheta_2\end{array}\right)\,\,\,,\qquad
&\begin{tikzpicture}[baseline=0.5cm]\draw (4.7,.35) arc (0:180:.5) [->];
\draw (4.72,0.7) node {\tiny $1$};\end{tikzpicture}
=\left(\begin{array}{cc}\varpi_2 & 0\end{array}\right)\,\,\,,\qquad\\
&\begin{tikzpicture}[baseline=0.5cm]\draw (7.5,.85) arc (180:360:.5) [->];
\draw (8.3,0.7) node {\tiny $1$};
\end{tikzpicture}=\left(\begin{array}{c}\varphi_2\\ 0\end{array}\right)\,\,\,,\qquad
&\begin{tikzpicture}[baseline=0.5cm]\draw (4.7,.35) arc (0:180:.5) [<-]; \end{tikzpicture}
=\left(\begin{array}{cc}\vartheta_3 & 0\end{array}\right)\,\,\,,\qquad
&\begin{tikzpicture}[baseline=0.5cm]\draw (7.5,.85) arc (180:360:.5) [<-];
\end{tikzpicture}=\left(\begin{array}{c}\vartheta_4\\0\end{array}\right)\,\,\,,\qquad
\end{array}
\end{equation}
and
\begin{align}\label{e:29.1}\tag{\ref{e:29}.1}
\begin{tikzpicture}[baseline=0.5cm]\draw (4.7,.35) arc (0:180:.5) [->];
\draw (4.72,0.7) node {\tiny $2$};\end{tikzpicture}
=\left(\begin{array}{cc}0 & \varpi_2\end{array}\right)\,\,\,,\qquad
\begin{tikzpicture}[baseline=0.5cm]\draw (7.5,.85) arc (180:360:.5) [->];
\draw (8.3,0.7) node {\tiny $2$};
\end{tikzpicture}=\left(\begin{array}{c}0 \\ \varphi_2\end{array}\right)\,\,\,,\qquad
\end{align}
\end{itemize}
The $2$-morphisms of ${\mathcal G}'$ satisfy the local relations (\ref{local1}), (\ref{local2.2}),
(\ref{local3}), and (\ref{local4.1}), where the local relation (\ref{local2.2}) is
\begin{align}\label{local2.2}\tag{\ref{local2}$''$}
\begin{tikzpicture}[>=stealth,baseline=0.8cm]
\draw (0,0) .. controls (0.8,0.8) .. (0,1.6)[<-];
\draw (0.8,0) .. controls (0,0.8) .. (0.8,1.6)[->];
\draw [shift={+(0.5,0)}](1.2,0.8) node {=};
\draw [shift={+(1,0)}](1.6,0) --(1.6,1.6)[<-];
\draw [shift={+(1,0)}](2.4,0) -- (2.4,1.6)[->];
\draw [shift={+(1.5,0)}](2.8,0.8) node {$-$};
\draw [shift={+(2,0)}](3.2,1.4) arc (180:360:.4);
\draw [shift={+(2.3,0.1)}] (3.6,1.2)  node {$1$};
\draw [shift={+(2,0)}](3.2,1.6) -- (3.2,1.4) ;
\draw [shift={+(2,0)}](4,1.6) -- (4,1.4) [<-];
\draw [shift={+(2,0)}](4,.2) arc (0:180:.4) ;
\draw [shift={+(2.3,-0.7)}] (3.6,1)  node {$1$};
\draw [shift={+(2,0)}](4,0) -- (4,.2) ;
\draw [shift={+(2,0)}](3.2,0) -- (3.2,.2) [<-];
\draw [shift={+(2.5,0)}](4.4,0.8) node {$-$};
\draw [shift={+(3.2,0.7)}] (5.2,0.6) node {$2$};
\draw [shift={+(3,0)}](4.8,1.4) arc (180:360:.4);
\draw [shift={+(3,0)}](4.8,1.6) -- (4.8,1.4) ;
\draw [shift={+(3,0)}](5.6,1.6) -- (5.6,1.4) [<-];
\draw [shift={+(3,0)}](5.6,.2) arc (0:180:.4);
\draw [shift={+(4.8,-0.7)}] (3.6,1)  node {$2$};
\draw [shift={+(3,0)}](5.6,0) -- (5.6,.2) node{\qquad .} ;
\draw [shift={+(3,0)}](4.8,0) -- (4.8,.2) [<-];
\end{tikzpicture}
\end{align}
These local relations (\ref{local2.2}), (\ref{local3}), and (\ref{local4.1})
imply that $QP\cong PQ\oplus {\mathbf 1}\oplus {\mathbf 1}$, which is
illustrated in the category ${\mathcal G}'$ as:
\begin{equation}\label{e:30}
\xy
(0,20)*+ {QP}="T1";
(-25,0)*+ {PQ}="m1";
(0,0)*+ {\mathbf 1}="m2";
(25,0)*+{{\mathbf 1} \,\, .}="m3";
(0, -20)*+{QP}="T2";
 {\ar@/_0pc/_{ \begin{tikzpicture} \draw [shift={(0.7,0)}] [<-] (0,0) -- (0.4,0.4);
 \draw [shift={(0.7,0)}] [->] (0.4,0) -- (0,0.4);
 \end{tikzpicture}
  \xy
 \endxy}
  "T1";"m1"};
 {\ar@/_0.0pc/_{ \begin{tikzpicture} \draw (0.42,-.5) arc (-15:195:.25) [->];
 \draw (0.48,-0.3) node{\tiny $1$};
 \end{tikzpicture}
 \xy
  \endxy} "T1";"m2"};
 {\ar@/^0pc/^{ \begin{tikzpicture} \draw (0.3,-.5) arc (-15:195:.25) [->];
 \draw (0.4,-0.3) node {\tiny $2$};
  \end{tikzpicture}
  \xy
  \endxy} "T1";"m3"};
 {\ar@/_0pc/_{  \begin{tikzpicture} \draw [->] (0,0) -- (0.4,0.4);
 \draw [<-] (0.4,0) -- (0,0.4);\end{tikzpicture}
  \xy
  \endxy} "m1";"T2"};
 {\ar@/_0.0pc/_{ \begin{tikzpicture} \draw [shift={+(0,0.3)}](0,0) arc (180:360:.25)[->];
 \draw (0.35,0.25) node{\tiny $1$};
   \end{tikzpicture}
 \xy
    \endxy} "m2";"T2"};
 {\ar@/^0pc/^{ \begin{tikzpicture} \draw [shift={+(0,0.3)}](0,0) arc (180:360:.25)[->];
 \draw (0.35,0.25) node {\tiny $2$};
 \end{tikzpicture}
  \xy
    \endxy} "m3";"T2"};
\endxy
\end{equation}

Let ${\sim}^{{\mathcal G}'}$ denote the local relations (\ref{local1}), (\ref{local2.2}), (\ref{local3}),
and (\ref{local4.1}) in the category ${\mathcal G}'$, and $\tilde{{\mathcal R}}_{{\mathcal G}'}$ the ring
generated by the string diagrams in (\ref{e:29}) and (\ref{e:29.1}),
similar to $\tilde{\mathcal R}_{{\mathcal H}'}$.
Then ${\mathcal R}_{{\mathcal G}'}:=\tilde{\mathcal R}_{{\mathcal G}'}/{{\sim}^{{\mathcal G}'}}$
is the ring of $2$-morphisms of the category ${\mathcal G}'$.
\begin{prop}\label{prop:1}
There is an injective ring homomorphism
$\psi: \tilde{\mathcal R}_{{\H'}} \lra \tilde{\mathcal R}_{{\mathcal G}'}$.
\end{prop}
\begin{proof}
The map $\psi$ is given by sending the string diagrams in (\ref{e:26}) to those in (\ref{e:29}) one by
one in order. Notice that the two rings $\tilde{\mathcal R}_{{\H'}}$ and $\tilde{\mathcal R}_{{\mathcal G}'}$
are generated by the mentioned string diagrams through both the vertical composition and the
horizontal composition of $2$-morphisms, but the multiplications of both rings are the vertical
compositions. It is routine to verify that $\psi$ is a ring homomorphism and
it is clear that $\psi$ is injective.
\end{proof}

\begin{cor}\label{cor:1}
For any $1$-morphism $M=\bigoplus_{i=1}^k {\mcP}^{m_i}{\mcQ}^{n_i}$ of ${\mathcal H}'$,
where $n_i\ge 0$ and $m_i\ge 0$ are integers, let $\tilde{M}=\bigoplus_{i=0}^k P^{m_i} Q^{n_i}$
be the corresponding $1$-morphism of ${\mathcal G}'$. Then ${\EndHp}(M)\cong {\End}_{{\mathcal G}'}(\tilde{M})$
under the injective homomorphism $\psi$ of Proposition \ref{prop:1}.
\end{cor}

\begin{theorem}[Conjectured in \cite{K} by Khovanov]\label{theorem:main}
$\gamma: {H_{\Z}}\lra K_0(\H)$ is an isomorphism.
\end{theorem}
The proof of this theorem will be given in the next section.

\section{Proof of Theorem \ref{theorem:main}}\label{section:4}
In this section, we use the same notation $'\H'$ to denote the $2$-category with same set of
$0$-cells $\Z$, the set of $1$-cells consisting of those generated by $\sP$, $\sQ$, and $\mathbf 1$,
but without any dimension shiftings, and the same set of $2$-morphisms as that of $'\H'$ with
the same definition of degrees. In such a setting, the $\C$-vector space of morphisms
between two $1$-morphisms $M$ and $N$ is $\HomHpp(M,N)$,
which consists of the $2$-morphisms of all degrees, instead of $\HomHpp(M,N)_0$, which
consists of the $2$-morphisms between $M$ and $N$ of degree $0$ in Section \ref{section:2}.
The local relations of $2$-morphisms are also the same. Then we have the basic
relation
\begin{equation}\label{e:basic}
\sQ\sP\cong \sP\sQ\oplus {\mathbf 1}\oplus {\mathbf 1}
\end{equation}
which is described by:
\begin{align}\label{e:qpdecompp}\tag{\ref{e:qpdecomp}$'$}
\xy
(0,20)*+ {\sQ\sP}="T1";
(-25,0)*+ {\sP\sQ}="m1";
(0,0)*+ {\mathbf 1}="m2";
(25,0)*+{{\mathbf 1} \,\, .}="m3";
(0, -20)*+{\sQ\sP}="T2";
 {\ar@/_0pc/_{ \begin{tikzpicture} \draw [<-] (0,0) -- (0.4,0.4);
 \draw [->] (0.4,0) -- (0,0.4);\end{tikzpicture}
  \xy
 \endxy}
  "T1";"m1"};
 {\ar@/_0.0pc/_{ \begin{tikzpicture} \draw (0.3,-.5) arc (-15:195:.25) [->];\end{tikzpicture} \xy
  \endxy} "T1";"m2"};
 {\ar@/^0pc/^{ \begin{tikzpicture} \draw (0.3,-.5) arc (-15:195:.25) [->];
  \filldraw (0.175,-0.225) circle (1.5pt); \draw (0.29,-0.1) node{$v$}; \end{tikzpicture}
  \xy
  \endxy} "T1";"m3"};
 {\ar@/_0pc/_{  \begin{tikzpicture} \draw [->] (0,0) -- (0.4,0.4);
 \draw [<-] (0.4,0) -- (0,0.4);\end{tikzpicture}
  \xy
  \endxy} "m1";"T2"};
 {\ar@/_0.0pc/_{ \begin{tikzpicture} \draw [shift={+(0,0.3)}](0,0) arc (180:360:.25)[->];
  \filldraw (0.125,0.105) circle (1.5pt) +(0,0.25) node{$v$}; \end{tikzpicture}
 \xy
    \endxy} "m2";"T2"};
 {\ar@/^0pc/^{ \begin{tikzpicture} \draw [shift={+(0,0.3)}](0,0) arc (180:360:.25)[->];\end{tikzpicture}
  \xy
    \endxy} "m3";"T2"};
\endxy
\end{align}

We define a subcategory ${\mathcal S}'$ of  $'\H'$ as:
\begin{itemize}
\item {\bf objects:} all rational integers;
\item {\bf $1$-morphisms:} same as the objects of $'\H'$ without dimension shiftings;
\item {\bf $2$-morphisms:} the $\C$-vector space generated by the following string diagrams
as a subring of the ring of $2$-morphisms of $'\H'$:
\begin{equation}\label{e:32}
\begin{split}
&\id_{\sP}=\begin{tikzpicture}[baseline=0.25cm]\draw[->] (0,0) -- (0,0.7);\end{tikzpicture}\,\,\,\,,\qquad
\id_{\sQ}=\begin{tikzpicture}[baseline=0.25cm]\draw[<-] (0,0) -- (0,0.7);\end{tikzpicture}\,\,\,\,,\qquad
\pi_1=\begin{tikzpicture}[baseline=0.25cm]\draw [<-](0,0) -- (0.7,0.7);
\draw [->](0.7,0) -- (0,0.7);\end{tikzpicture}\,\,\,,\qquad
\phi_1=\begin{tikzpicture}[baseline=0.25cm]\draw [->](0,0) -- (0.7,0.7);
\draw [<-](0.7,0) -- (0,0.7);\end{tikzpicture}\,\,\,,\qquad
\theta_1=\begin{tikzpicture}[baseline=0.25cm]\draw [->](0,0) -- (0.7,0.7);
\draw [->](0.7,0) -- (0,0.7);\end{tikzpicture}\,\,\,,\qquad
\theta_2=\begin{tikzpicture}[baseline=0.25cm]\draw [<-](0,0) -- (0.7,0.7);
\draw [<-](0.7,0) -- (0,0.7);\end{tikzpicture}\,\,\,,\qquad \\
&\pi_{21}=\begin{tikzpicture}[baseline=0.5cm]\draw (4.7,.35) arc (0:180:.5) [->]; \end{tikzpicture}\,\,\,,\qquad
\phi_{21}=\begin{tikzpicture}[baseline=0.5cm]\draw (7.5,.85) arc (180:360:.5) [->];
\filldraw (8,0.35) circle (2pt)+(0.2,0) node {$v$};
\end{tikzpicture}\,\,\,,\qquad
\theta_3=\begin{tikzpicture}[baseline=0.5cm]\draw (4.7,.35) arc (0:180:.5) [<-]; \end{tikzpicture}\,\,\,,\qquad
\theta_4=\begin{tikzpicture}[baseline=0.5cm]\draw (7.5,.85) arc (180:360:.5) [<-];
\end{tikzpicture}\,\,\,,\qquad
\end{split}
\end{equation}
and
\begin{align}\label{e:32.1}\tag{\ref{e:32}.1}
\pi_{22}\,\,=\,\, \begin{tikzpicture}[baseline=0.5cm]\draw (8.5,.35) arc (0:180:.5) [->];
\filldraw (8.0,0.85) circle (2pt)+(0.2,0) node{$v$};\end{tikzpicture}\,\,\,,\qquad
\phi_{22} \,\,=\,\, \begin{tikzpicture}[baseline=0.5cm]\draw (7.5,.85) arc (180:360:.5) [->];
\end{tikzpicture}\,\,\,.
\end{align}
\end{itemize}
Let $\tilde{\mathcal R}_{{\mathcal S}'}$ be the $\C$-vector space generated by the diagrams
through horizontal composition of $2$-morphisms of those in (\ref{e:32}) and
(\ref{e:32.1}), modulo relative to boundary isotopies.
Then $\tilde{\mathcal R}_{{\mathcal S}'}$ is a ring with the obvious
addition and the multiplication being the vertical composition of $2$-morphisms,
and ${\mathcal R}_{{\mathcal S}'}:=\tilde{{\mathcal R}}_{{\mathcal S}'}/{{\sim}}^{{\mathcal S}'}$ is the
ring of $2$-morphisms of ${\mathcal S}'$, where ${\sim}^{{\mathcal S}'}$ is
the local relations given in (\ref{local1})--(\ref{local4}).

\begin{theorem}\label{theorem:6}
The categories ${\mathcal G}'$ and ${\mathcal S}'$ are isomorphic.
\end{theorem}
\begin{proof}
We define a functor $F:{\mathcal G}'\lra {\mathcal S}'$ by sending $n\in{\Z}$ to $n$ on $0$-cells,
sending ${\mathbf 1}$, $P$, and $Q$ of ${\mathcal G}'$ respectively to ${\mathbf 1}$, $\sP$, and $\sQ$
of ${\mathcal S}'$ on $1$-cells, and sending the string diagrams in (\ref{e:29}) and (\ref{e:29.1})
to those in (\ref{e:32}) and (\ref{e:32.1}) one by one in order.
Then $F$ is one to one on the $0$-cells and $1$-cells, and also preserves direct sums and
compositions of $1$-cells. It is clear that
${\tilde{\mathcal R}}_{{\mathcal G}'}\cong {\tilde{\mathcal R}}_{{\mathcal S}'}$ as $\C$-algebras, and
the local relations (\ref{local1}), (\ref{local2.2}), (\ref{local3}), and (\ref{local4.1})
of the category ${\mathcal G}'$ are preserved under $F$, therefore we get the
isomorphism ${\mathcal R}_{{\mathcal G}'}\cong {\mathcal R}_{{\mathcal S}'}$ between the ring of
$2$-morphisms of ${\mathcal G}'$ and that of ${\mathcal S}'$.
\end{proof}

The grading of $2$-morphisms of ${\H}'$ in Section \ref{section:2} restricts to the $2$-morphisms
of ${\mathcal S}'$, with respect to both $deg$ and $sdeg$, thus $\Hom_{{\mathcal S}'}(M,N)$ is a
graded $\C$-vector space for any $1$-morphisms $M$ and $N$ of ${\mathcal S}'$, and
$\End_{{\mathcal S}'}(M)$ is a graded $\C$-algebra.
Let $\mathcal S$ be the Karoubi envelope of ${\mathcal S}'$
and $\mathcal G$ the Karoubi envelope of ${\mathcal G}'$.

\begin{cor}\label{kariso}
The $2$-categories $\mathcal G$ and $\mathcal S$ are isomorphic.
\end{cor}
We need to study the Grothendieck ring $K_0({\mathcal S})=K_0(\mathcal G)$.

From Equation (\ref{e:basic}), we see that any $1$-morphism of ${\mathcal S}'$ is isomorphic to
\begin{equation}\label{e:34}
M_{\m,\n}=\bigoplus_{i=1}^k \sP^{m_i}\sQ^{n_i}
\end{equation}
where $\m=(m_1,m_2,\cdots, m_k)$ and $\n=(n_1,n_2,\cdots, n_k)$ are $k$-tuples of non-negative integers.
Therefore a $1$-morphism of $\mathcal S$ is of the form $(M_{\m,\n}, e)$ with
$e\in R_{\m,\n}:={\End}_{{\mathcal S}'}(M_{\m,\n})$ an idempotent element.

\begin{prop}\label{prop:2}
For non-negative integers $m,n,m_i,n_i$, $i=1,2$,
\begin{itemize}
\item[(1).] $\Hom_{{\mathcal S}'}(\sP^{m_1}\sQ^{n_1}, \sP^{m_2}\sQ^{n_2})\ne 0$ if and only if
$n_2-m_2=n_1-m_1$;
\item[(2).] If $(m_1,n_1)\ne (m_2,n_2)$, then $sdeg(\alpha)>0$ for any $\alpha\in
\Hom_{{\mathcal S}'}(\sP^{m_1}\sQ^{n_1}, \sP^{m_2}\sQ^{n_2})$;
\item[(3).] The $\C$-algebra $R_{m,n}:=\End_{{\mathcal S}'}(\sP^m \sQ^n)$ is graded according to
the degree $sdeg$:
\[
R_{m,n}=\bigoplus_{i=0}^\infty R^i_{m,n}
\]
where $R_{m,n}^i=\{\alpha\in R_{m,n}\,\, |\,\,\text{$\alpha$ is homogeneous with respect to $vdeg$ and }
 vdeg(\alpha)=i/2\}$ is a finite dimensional $\C$-vector
space for $i\ge 0$, they satisfy
\[
R_{m,n}^i\cdot R_{m,n}^j\subseteqq R_{m,n}^{i+j},
\]
and $R_{m,n}^0\cong {\C}[S_m]\otimes {\C}[S_n]$.
\end{itemize}
\end{prop}
\begin{proof}
The $2$-morphisms of ${\mathcal S}'$ are generated by the string diagrams in (\ref{e:32}) and
(\ref{e:32.1}) through horizontal composition, vertical composition, and $\C$ linearity. Among
the string diagrams (that is, the $2$-morphisms generated from (\ref{e:32}) and
(\ref{e:32.1}) through horizontal position and vertical composition)
we consider in this proposition, only those containing clockwise cup or
clockwise cap have shifted degree bigger than $0$, these $2$-morphisms are from $\sP^m\sQ^n
\to \sP^{m+l}\sQ^{n+l}$ for some integer $l$ and others are from $\sP^m\sQ^n
\to \sP^{m}\sQ^{n}$. Therefore (1) and (2) hold. The homogeneous $2$-morphisms of shifted degree $0$
are generated by the crossings $\theta_1$ and $\theta_2$ in (\ref{e:32}), hence
$R^0\cong {\C}[S_m]\otimes {\C}[S_n]$.
\end{proof}

An element $A\in R_{\m,\n}$ can be written as
\[
A=\left(
\begin{array}{cccc}
a_{1,1} & a_{1,2} & \cdots & a_{1,k}\\
a_{2,1} & a_{2,2} & \cdots & a_{2,k}\\
\cdots  & \cdots  & \cdots & \cdots\\
a_{k,1} & a_{k,2} & \cdots & a_{k,k}
\end{array}
\right)
\]
with $a_{i,j}\in \Hom_{{\mathcal S}'}(\sP^{m_i}\sQ^{n_i}, \sP^{m_j}\sQ^{n_j})$. We define a filtration
on $R_{\m,\n}$ by
\begin{equation}\label{filtration}
R_{\m,\n}=F_0\supsetneqq F_1\supset F_2\supset \cdots \supset F_p\supset F_{p+1}\supset \cdots
\end{equation}
satisfying $F_p\cdot F_q\subset F_{p+q}$ for $p,q\ge 0$ and $\bigcap_{p=0}^\infty F_p=0$, where
\[
F_p=\{A=(a_{i,j})\in R_{\m,\n}\,\, |\,\, sdeg(a_{i,j})\ge p/2 \text{ for $1\le i,j\le k$}\}.
\]
Moreover, $F_1$ is a two-sided ideal of the unital ring $R_{\m,\n}$, and there is the split
exact sequence:
\begin{equation}\label{splitexact}
0\lra F_1\lra R_{\m,\n}\lra R_{\m,\n}^0\lra 0
\end{equation}
where
\[
R_{\m,\n}^0=\{A=(a_{i,j})\in R_{\m,\n}\,\, |\,\, \text{ $a_{i,j}$ is homogeneous with
$vdeg(a_{i,j})=0$ for $1\le i,j\le k$}\}
\]
is a subalgebra of $R_{\m,\n}$.

\begin{lemma}[\cite{WWY}]\label{lemma:1}
Let $R$ be a unital ring containing $\Z$ with a filtration
\[
R=F_0\supsetneqq F_1\supset F_2\supset \cdots \supset F_p\supset F_{p+1} \supset\cdots
\]
such that $F_p\cdot F_q\subset F_{p+q}$ for integers $p,q\ge 0$ (therefore $F_1$ is a two-sided
ideal of $R$) and $\bigcap_{p=0}^\infty F_p = 0$. Suppose that $R_0$ is a unital subring of $R$ such that
$R_0\cong R/F_1$ and the exact sequence
\[
0\lra F_1\lra R\lra R_0\lra 0
\]
is split. Then $K_0(R)\cong K_0(R_0)$ as abelian groups.
\end{lemma}

\begin{cor}\label{cor:2}
Let $R_{\m,\n}$ and $R^0_{\m,\n}$ be as above, and
\[
X=\{(m_i,n_i)\,\,|\,\, \m=(m_1,m_2,\cdots,m_k), \n=(n_1,n_2,\cdots,n_k), 1\le i\le k\}
\]
the set of different tuples $(m_i,n_i)$ appearing in the decomposition of $M_{\m,\n}$ in (\ref{e:34}).
Then
\begin{equation}\label{e:36}
K_0(R_{\m,\n})\cong \bigoplus_{(m,n)\in X} K_0(\End_{{\mathcal S}'}(\sP^m\sQ^n))
\cong \bigoplus_{(m,n)\in X} K_0(\C[S_m]\otimes {\C}[S_n]).
\end{equation}
\end{cor}
\begin{proof}
By (\ref{filtration}) and (\ref{splitexact}), we can apply Lemma \ref{lemma:1}, hence
$
K_0(R_{\m,\n})\cong K_0(R^0_{\m,\n}).
$
Proposition \ref{prop:2} implies that
\[
K_0(R^0_{\m,\n})\cong \bigoplus_{(m,n)\in X} K_0(M_{s(m,n)}(R^0_{m,n}))
\]
where $s(m,n)=\#\{i\,\, | \,\, (m_i,n_i)=(m,n) \text{ for $1\le i\le k$}\}$ and
$M_{s(m,n)}(R^0_{m,n})$ denotes the matrix ring over the ring $R^0_{m,n}$ of rank
$s(m,n)$. By the ``Morita invariance" theorem (see Chapter 1 of \cite{Ros}),
$K_0(M_{s(m,n)}(R^0_{m,n}))\cong K_0(R^0_{m,n})$, and from Part 3 of Proposition \ref{prop:2},
we have $R^0_{m,n}\cong {\C}[S_m]\otimes {\C}[S_n]$. Therefore the second isomorphism
of (\ref{e:36}) holds.
\end{proof}

\begin{prop}\label{prop:3}
Let $\sP_m=(\sP^n, e_{(n)})$ and $\sQ_n=(\sQ^n, e_{(n)})$ for $n\in {\Z}$ non-negative.
The Grothendieck ring $K_0({\mathcal S})$ is generated by the classes
$[\sP_n],\, [\sQ_n]$ for $n=0, 1, 2, \cdots$ as an algebra over $\Z$.
\end{prop}
\begin{proof}
Any element of $K_0({\mathcal S})$ can be written as $[(M_{\m,\n}, e)]$, where $M_{\m,\n}$ is
given in (\ref{e:34}) and $ e\in R_{\m,\n}=\End_{{\mathcal S}'}(M_{\m,\n})$ is an
idempotent. Corollary \ref{cor:2} implies that $[(M_{\m,\n}, e)]$ is a ${\Z}$-linear
combination of the terms $[(\sP^m,e_{\lambda})]\cdot [(\sQ^n, e_{\mu})]$
as elements of $K_0({\C}[S_m]\otimes {\C}[S_n])$, where $\lambda$
is a partition of the integer $m$ and $\mu$ is a partition of the integer $n$.
As the inductions of the trivial representations of $S_{n_1}\times S_{n_2}\times \cdots
\times S_{n_l}\hookrightarrow S_n$ for $n_1+n_2+\cdots n_l=n$ generate all virtual
representations of $S_n$ (see \cite{FH} or \cite{M}), we see that
\[[(\sP^m, e_{\lambda})]=\sum_{i_1+i_2+\cdots +i_l=m} c_{i_1,i_2,\cdots, i_l}[\sP_{i_1}]\cdot
[\sP_{i_2}]\cdots [\sP_{i_l}]
\]
with $c_{i_1,i_2,\cdots, i_l}\in {\Z}$ and a similar formula holds for $[\sQ^n,e_{\mu}]$.
Therefore the Grothendieck ring $K_0({\mathcal S})$ is generated
by $[\sP_n], [\sQ_n]$ for $n=0,1,2,\cdots$.
\end{proof}

\noindent{\bf Proof of Theorem \ref{theorem:main}}. By Theorem \ref{theorem:4} of Khovanov, we
need only prove the surjectivity of $\gamma$, that is, for any given
$1$-morphism $(M,e)$ of ${\mathcal H}$, where $M$ is a $1$-morphism of ${\H}'$ and
$e\in {\EndHp}(M)$ is
an idempotent $2$-morphism, the class $[(M,e)]$ can be written as a $\Z$-linear
combination of $[{\mathcal P}_{a_1}]\cdot\cdots\cdot
[{\mathcal P}_{a_i}]\cdot [{\mathcal Q}_{b_1}]\cdot\cdots\cdot [{\mathcal Q}_{b_j}]$'s,
with $a_1, \cdots, a_i$, and $b_1, \cdots, b_j$ equal to non-negative integers.

Due to the decomposition ${\mathcal Q}{\mathcal P}\cong {\mathcal P}{\mathcal Q}\oplus {\mathbf 1}$
in the category $\H'$, the $1$-morphism $M$ can be expressed as
\[
M\cong \bigoplus_{i=1}^k {\mathcal P}^{m_i} {\mathcal Q}^{n_i}.
\]
Under the injective ring homomorphism
$\psi: {\tilde{\mathcal R}}_{\H'}\lra {\tilde{\mathcal R}}_{{\mathcal G}'}$ in Proposition \ref{prop:1},
we have the isomorphism $\EndHp(M)\cong {\End}_{{\mathcal G}'}(\tilde{M})$ in Corollary \ref{cor:1}, where
$\tilde{M}=\bigoplus_{i=1}^k P^{m_i} Q^{n_i}$. The $2$-category ${\mathcal G}'$ is
isomorphic to ${\mathcal S}'$ by Theorem \ref{theorem:6}. Under this isomorphism, $\tilde{M}$
corresponds to
\[
M_{\m,\n}=\bigoplus_{i=1}^k \sP^{m_i} \sQ^{n_i}
\]
and we get
\[
{\EndHp}(M)\xrightarrow{\cong} {\End}_{{\mathcal G}'}(\tilde{M})\xrightarrow{\cong}
{\End}_{{\mathcal S}'}(M_{\m,\n}) = R_{\m,\n}.
\]
After taking the Grothendieck groups in the above isomorphisms and applying
Proposition \ref{prop:3}, we see that the class $[(M,e)]$ can be written as a $\Z$-linear
combination of $[{\mathcal P}_{a_1}]\cdot\cdots\cdot
[{\mathcal P}_{a_i}]\cdot [{\mathcal Q}_{b_1}]\cdot\cdots\cdot [{\mathcal Q}_{b_j}]$'s,
with $a_1, \cdots, a_i$, and $b_1, \cdots, b_j$ equal to non-negative integers.
\hfill $\square$

\begin{Remark}\label{remark:proofmain}
We can also use the same method as in \cite{WWY} to prove the injectivity of  $\gamma$.
Let $H=H_{\Z}\otimes {\C}$ be the Heisenberg algebra with coefficients in the
complex number field $\C$. Then $H$ is generated by $\{a_n\}_{n\in \Z}$ with the defining
relation (\ref{e:1.1}), or by $\{p_n, q_n\}_{n\in {\mathbb N}}$ with the defining relations
(\ref{e:1.4})--(\ref{e:1.6}), in accordance with (\ref{e:pqrelation}). An element of
$H_{\Z}$ is mapped to $K_0(\H)$ by $\gamma$, which can be viewed as an
endomorphism of the Fock space representation of $H$ constructed from the generators
$a_n$ of $H$, by the work in \cite{FJW1} and \cite{FJW2}. As the Fock space representation is
faithful, we get the injectivity of $\gamma$, see \cite{WWY} for detail.
\end{Remark}

\section{Further Discussion}\label{section:5}
The category ${\mathcal G}'$ and its Karoubi envelope $\mathcal G$ are constructed from
the category $\H'$ by requiring that ${\mathcal G}'\ni P= \mcP\in {\H'}$ and
${\mathcal G}'\ni Q=\mcQ\oplus \mcQ\in{\H'}$ on the $1$-morphisms, and
string diagrams be given in (\ref{e:29}) and (\ref{e:29.1})
to construct all the $2$-morphisms by horizontal compositions and vertical compositions.
As a subcategory of $'\H'$,  the category ${\mathcal S}'$ is basically the same as $'\H'$:
the missing $1$-morphisms from $'\H'$ are the dimension shiftings of the $1$-morphisms of
${\mathcal S}'$, the missing $2$-morphisms from $'\H'$ are some of those with dot $2$-morphisms.
Theorem \ref{theorem:6} says that the two categories ${\mathcal G}'$ and ${\mathcal S}'$ are
isomorphic. From such reasoning, the categorification $'\H$ of the deformed Heisenberg algebra
$'H$ is obtained from the categorification $\H$ of $H$ as  the $1$-morphism $\sQ$
of $'\H$ ``splits'' into two parts: $\sQ\cong \sQ_1\oplus \sQ_2$. We explain this
from the point of view of representation theory, with the idea based on the Fock space
representations of the Heisenberg categorifications in \cite{K} and \cite{CL}.

Let $A={\C}[x]$ be the polynomial ring of one variable over $\C$, $A_0={\C}$, and for
an integer $n\ge 1$
\[
A_n=\underbrace{(A\otimes A\otimes \cdots \otimes A)}_{\text{$n$ times}}\rtimes S_n
\]
be the ${\C}$-algebra with the obvious addition and the multiplication given by
\[
(f,\sigma)\cdot (g,\tau):= (f\cdot \sigma(g),\sigma\tau)
\]
for $f,g\in A\otimes A\otimes\cdots\otimes A$, and $\sigma,\tau\in S_n$, where $\sigma(g)$
denotes the action of $S_n$ on $A\otimes A\otimes \cdots\otimes A$ by the permutation of
$\sigma$ on the $n$ tensor components of $g$.  Let $D(\An\gmod)$
be the bounded derived category of graded (left) $\An$-modules.
Define the category $\mathcal F$ by
\[
{\mathcal F}:= \bigoplus_{n=0}^\infty D(\An\gmod)
\]
with morphisms as the functors among the categories $D(\An\gmod)$.
There are endofunctors $\mathbf 1$, $\sP$, $\sQ$ of $\mathcal F$ given in the following.
\begin{itemize}
\item $\mathbf 1$ is the identity endofunctor of $\mathcal F$.
\item $\sP: {\mathcal F}\to {\mathcal F}$ is defined by
\[
\begin{array}{rrll}
\sP(n): &D(\An\gmod)&\to     &D(A_{n+1}\gmod)\\
        &M          &\mapsto &A_{n+1}\otimes_{\An\otimes A}(\An\otimes{\C})\otimes_{\An} M
\end{array}
\]
where we use the natural inclusion $\An\otimes A\hookrightarrow A_{n+1}$ from the
embedding of the group $S_n\cong S_n\times S_1\hookrightarrow S_{n+1}$ in the tensor product.
\item Similarly, $\sQ: {\mathcal F}\to {\mathcal F}$ is defined by
\[
\begin{array}{rrll}
(n)\sQ: &D(A_{n+1}\gmod) & \to  & D(\An\gmod)\\
        &M &\mapsto &(\An\otimes {\C})\otimes_{\An\otimes A}A_{n+1}\otimes_{A_{n+1}} M.
\end{array}
\]
\end{itemize}
In the above definition of the endofunctors $\sP$ and $\sQ$, the values $\sP(f)$ and
$\sQ(g)$ for morphisms $f,g$ of $\mathcal F$ are defined in the obvious way.

There is the free $A$-module resolution of $\C$:
\begin{equation}\label{resofC}
0\lra xA\lra A\lra {\C}\lra 0.
\end{equation}
Hence we get the free ${A\otimes A\otimes\cdots\otimes A}$ ($n$ times) module
resolution of $\C$ for each $n\ge 0$
\[
0\lra xA\otimes xA\otimes\cdots\otimes xA\lra \bigoplus_{i=1}^n xA\otimes\cdots\otimes A\otimes\cdots\otimes xA
\lra\cdots\lra A\otimes A\otimes\cdots\otimes A\lra {\C}\lra 0
\]
by tensoring the free $A$-module resolution (\ref{resofC}) of $\C$ for $n$ times, and this gives rise to
a free resolution of $\C[S_n]$ as an $A_n$-module:
\begin{equation}\label{e:38}
0\lra P_n\lra P_{n-1}\lra \cdots \lra P_1\lra P_0\lra {\C}[S_n]\lra 0,
\end{equation}
where each $P_i$ is a free $A_n$-module. Moreover, for each partition $\lambda$ of $n$, let $e_{\lambda}$
be the idempotent in ${\C}[S_n]$ associated to $\lambda$, then the corresponding irreducible representation
of $S_n$ is ${\C}[S_n]e_{\lambda}$, and we have the projective $A_n$-module resolution of
${\C}[S_n]e_{\lambda}$:
\begin{equation}\label{e:39}
0\lra P^{\lambda}_n\lra P^{\lambda}_{n-1}\lra \cdots\lra P^{\lambda}_1\lra P^{\lambda}_0\lra
{\C}[S_n]e_{\lambda}\lra 0.
\end{equation}

Through the natural mapping $A\lra {\C}\cong A/xA$, a ${\C}[S_n]$-module is equipped with the natural
$A_n$-module structure in which the action of $x$ is $0$. We regard the category ${\C}[S_n]\dmod$ as the
full subcategory of $D(A_n\gmod)$ concentrated at $0$, and let
\[
\text{${\mathcal C}_n = $ the full
subcategory of $D(A_n\gmod)$ with objects given by \,\,\,}
\bigoplus_{i=1}^k \,M_i\,[l_i],
\]
where each $M_i$ is an object of ${\C}[S_n]\dmod$. Then we put
\[
\bar{{\mathcal F}}:= \bigoplus_{n\ge 0} \,\,\,{\mathcal C}_n.
\]
The functors $\mathbf 1$, $\sP$, $\sQ$ restrict to $\bar{\mathcal F}$, and we can also define the
natural transformations between the functors generated by $\mathbf 1$, $\sP$, $\sQ$, as in
\cite{WWY} or in \cite{CL} for more general setting.

\begin{prop}\label{prop:4}
The endofunctor $\sQ$ of $\bar{\mathcal F}$ breaks into two parts: $\sQ=\sQ_1\oplus \sQ_2$ such
that they satisfy
\begin{align}
&\sQ_1\sP\cong \sP\sQ_1\oplus {\mathbf 1}, \label{e:40}\\
&\sQ_2\sP\cong \sP\sQ_2\oplus {\mathbf 1}\lba 1\rba. \label{e:41}
\end{align}
\end{prop}
\begin{proof}
On the category ${\mathcal C}_0$ of finite dimensional vector space over $\C$, it's easy to see
that for $V_0=\C$
\[
\sP(V_0)= \text{ the complex in ${\mathcal C}_1: 0\lra xA\lra A\lra 0$,}\quad
\text{ and $\sQ(V_0)=0$ by definition}.
\]

On the category ${\mathcal C}_1$, the object $V_1:=\C$ in the category ${\mathcal C}_1$ is
represented by its resolution (\ref{resofC}). It is clear that $\sP(V_1)$ is the complex
in ${\mathcal C}_2$:
\[
0\lra P_2\lra P_1 \lra P_0\lra 0
\]
with $P_2, P_1, P_0$ given as
\begin{align*}
&P_2=(xA\otimes xA)\rtimes S_2, \quad
P_1=(xA\otimes A)\rtimes S_2 \oplus (A\otimes xA)\rtimes S_2, \quad
P_0=(A\otimes A)\rtimes S_2,
\end{align*}
and $\sQ(V_1)$ is the complex in ${\mathcal C}_0$:
\[
0\lra {\C}x \xrightarrow{d_1} {\C} \xrightarrow{d_0} 0,
\]
where both $d_1$ and $d_0$ are zero maps, thus
\[
\sQ(V_1)\cong V_0\oplus V_0\lba 1\rba, \quad\text{ where $V_0={\C}$.}
\]

On the category ${\mathcal C}_n$ for $n\ge 1$, let $V_{\lambda}$ be the irreducible representation
of $S_n$ associated to a given partition $\lambda$ of $n$. Then we know that
\begin{equation}\label{e:42}
{\C}[S_{n+1}]\otimes_{{\C}[S_n]} V_{\lambda}\cong \bigoplus_
{\substack{\mu\succ\lambda,\,\\ |\mu|=n+1}}V_{\mu},
\quad \text{ and }\quad
{\rm Res}^{S_n}_{S_{n-1}} V_{\lambda}\cong \bigoplus_
{\substack{\mu\prec\lambda,\\|\mu|=n-1}} V_{\mu}.
\end{equation}
Write the irreducible representation $V_{\lambda}$ of $S_n$ as a complex in ${\mathcal C}_n$
and the irreducible representations $V_{\mu}$ of $S_{n+1}$ (or $S_{n-1}$) as complexes in
${\mathcal C}_{n+1}$ (or ${\mathcal C}_{n-1}$) by using the resolution (\ref{e:39}), we can deduce
that
\begin{equation}\label{e:43}
\sP(V_{\lambda})\cong \bigoplus_
{\substack{\mu\succ\lambda,\,\\ |\mu|=n+1}}V_{\mu},\quad\text{ and } \quad
\sQ(V_{\lambda})\cong \bigoplus_
{\substack{\mu\prec\lambda,\\|\mu|=n-1}} V_{\mu}\,\, \bigoplus \bigoplus_
{\substack{\mu\prec\lambda,\\|\mu|=n-1}} V_{\mu}\lba 1\rba,
\end{equation}
in the category $\bar{\mathcal F}$. We define two endofunctors $\sQ_1$ and $\sQ_2$ of
$\bar{\mathcal F}$ as
\[
\sQ_1(n),\,\sQ_2(n): {\mathcal C}_n\lra {\mathcal C}_{n-1}
\]
by
\[
\sQ_1(n)(V_{\lambda}):= \bigoplus_
{\substack{\mu\prec\lambda,\\|\mu|=n-1}} V_{\mu}, \quad\text{ and }\quad
\sQ_2(n)(V_{\lambda}):= \bigoplus_
{\substack{\mu\prec\lambda,\\|\mu|=n-1}} V_{\mu}\lba 1\rba,
\]
in accordance with the second isomorphism of (\ref{e:43}). The isomorphisms (\ref{e:40})
and (\ref{e:41}) now follow the above definition of $\sQ_1$ and $\sQ_2$.
\end{proof}

\begin{cor}\label{cor:4}
There is a representation of $\H$ which sends all $0$ cells to $\bar{\mathcal F}$,
sends the $1$-cells $\mathcal P$ to $\sP$ and ${\mathcal Q}$ to $\sQ_2$.
\end{cor}
\begin{proof}
We need only define the natural transformations between the endofunctors of $\bar{\mathcal F}$
generated by $\sP$ and $\sQ_2$ in a similar way as in \cite{WWY}, and the dimension shiftings
on the endofunctors of $\bar{\mathcal F}$ can be omitted and replaced by some degrees on the
natural transformations.
\end{proof}
By the above corollary, we can give an alternative proof of the surjectivity of $\gamma$ in
Theorem \ref{theorem:main}. The dimension shifting on $V_{\mu}$ in the definition of
$\sQ_2(V_{\lambda})$ gives rise to degrees on the natural transformations, hence filtrations
like (\ref{filtration}), therefore the proof in Section \ref{section:4} goes through.

All the above discussion can be generalized to the case in which $A={\C}[x_1,x_2,\cdots,x_m]$
is the polynomial ring over $\C$ of $m$ variables, then the resolution (\ref{resofC}) is
replaced by the Koszul resolution
\[
0\lra \wedge^m (A^m)\lra \cdots\lra \wedge^2(A^m)\lra A^m\lra A\lra A/I\cong{\C}\lra 0
\]
where $I=(x_1,x_2,\cdots,x_m)A$. Such a setting would lead to categorifications of Heisenberg
algebras with possible geometric bearings to the framed moduli space of torsion free sheaves
on surfaces, which we will give detailed discussion somewhere else.

\end{document}